\documentclass[aps,pra,10pt,twocolumn,superscriptaddress,floatfix,nofootinbib,showpacs,longbibliography, groupedaddress]{revtex4-1}

\usepackage{comment}
\usepackage{tabularx}
\usepackage{ragged2e}

\usepackage[utf8]{inputenc}  
\usepackage[T1]{fontenc}     
\usepackage[table]{xcolor}
\usepackage[british]{babel}  
\usepackage[sc,osf]{mathpazo}
\usepackage[scaled=0.86]{berasans}  
\usepackage[colorlinks=true, citecolor=blue, urlcolor=blue]{hyperref}  
\usepackage{graphicx} 
\usepackage[babel]{microtype}  
\usepackage{amsmath,amssymb,amsthm,bm,amsfonts,mathrsfs,bbm} 

\usepackage{xspace}  
\usepackage{pgf,tikz}
\usepackage{xcolor}
\usepackage{multirow}
\usepackage{array}
\usepackage{bigstrut}
\usepackage{braket}
\usepackage{color}
\usepackage{natbib}
\usepackage{multirow}
\usepackage{mathtools}
\usepackage{float}
\usepackage[caption = false]{subfig}
\usepackage{xcolor,colortbl}
\usepackage{color}
\usepackage{rotating}
\usepackage{tikz}
\usetikzlibrary{quantikz}

\newcommand{\be}{\begin{equation}}
\newcommand{\ee}{\end{equation}}
\newcommand{\ba}{\begin{eqnarray}}
\newcommand{\ea}{\end{eqnarray}}

\newtheorem{theorem}{Theorem}
\newtheorem{corollary}{Corollary}
\newtheorem{definition}{Definition}

\def\>{\rangle}
\def\<{\langle}

\usepackage{centernot}
\usepackage{subfig}

\begin{document}

\title{Resource dependent undecidability: computability landscape of distinct Turing theories}

\author{Airin Antony}
\affiliation{School of Physics, IISER Thiruvananthapuram, Vithura, Kerala 695551, India.}


\begin{abstract}
Can a problem undecidable with classical resources be decidable with quantum ones? The answer expected is no; as both being Turing theories, they should not solve the Halting problem - a problem unsolvable by any Turing machine. Yet, we provide an affirmative answer to the aforesaid question. We come up with a novel logical structure to formulate infinitely many such problems for any pair of distinct Turing theories, including but not limited to the classical and quantum theories. Importantly, a class of other decision problems, such as the Halting one, remains unsolvable in all those theories. The apparent paradoxical situation gets resolved once it is perceived that the reducibility of Halting problem changes with varying resources available for computations in different theories. In the end, we propose a multi-agent game where winnability of the player having access to only classical resources is undecidable while quantum resources provide a perfect winning strategy.
\end{abstract}

\maketitle
\section{Introduction}
The most astonishing results in quantum information theory have consistently demonstrated the striking advantages of quantum resources over their classical counterparts in different tasks \cite{Bennett92,Deutsch1992,Bennett93,Shor94,Grover96,Cleve1998,Dale15}. For instance, recently, a specific computational problem has been shown to be solvable more efficiently by a quantum algorithm than any futuristic classical algorithm \cite{Raz19}. Moving to the paradigm of computability, here we present one such counter-intuitive extreme scenario. For any two distinct physical theories, including but not limited to the classical and quantum ones, we show that infinitely many decision problems can be constructed that are undecidable in one theory while decidable in the other, although a class of other decision problems remain undecidable in both.

Undecidable problems are the decision problems (yes-or-no questions on an infinite set of inputs) that no algorithm, or Turing machine, can solve in finite time for all inputs \cite{undecidabilitybook}. The celebrated `Halting problem' ($\mathcal{H}$), which asks whether a given program eventually halts on a given input, was shown to be undecidable by Turing as early as in $1936$ \cite{Turing37}. A plethora of other mathematical problems soon followed, their undecidability demonstrated via reduction: them being decidable would imply that $\mathcal{H}$ is decidable, a contradiction \cite{Poonen, davisundecidableexamples}. Lately,  undecidability has garnered more attention among physicists, as certain long-standing problems in physics with foundational as well as real-world implications were finally proved to be undecidable \cite{Deutsch85,Wolfram85,Costa91,Cubitt15,Slofstra2019,Fritz2020,Bausch20,Shiraishi21,scandi2021}. One such recent development involves identifying the decision problems perfectly decidable in classical scenarios but undecidable in analogous quantum settings. The pioneering effort in this direction by Eisert-Müller-Gogolin demonstrated that the quantum measurement occurrence problem (QMOP) is undecidable despite the decidability of its classical analogue, the classical measurement occurrence problem (CMOP) \cite{Eisert12}. Similarly, a decidable problem concerning goal POMDPs (partially observable Markov decision processes) was found to be  undecidable for goal QOMDPs, their quantum analogues \cite{Barry14}. Such approaches are quite fascinating from a physical computation perspective. The classical and quantum theories are widely believed to qualify for what we shall define as the Turing theories and cannot solve undecidable problems such as $\mathcal{H}$ \cite{Cotogno03,Galton06,ARRIGHI12}. Nonetheless, in principle, there exist logically well-defined but abstract ‘oracle theories’ capable of solving $\mathcal{H}$ \cite{Abramson71,Stewart91,Copeland02}. With access to seemingly unphysical and unattainable resources, converting a problem from undecidable to decidable is less appealing. What makes Eisert {\it et al}'s approach interesting is that one such undecidable problem defined within the constructs of a standard Turing theory (quantum) had a perfectly decidable counterpart within another (classical).

Here, we stretch these concepts to their extremes by enquiring whether similar results hold for the same decision problem rather than for analogous ones. Specifically, consider the following scenario: within the constructs of some Turing theory $\mathrm{R}_1$ (say, the classical theory), let the problem $\mathcal{H}$ be reducible to another problem $\mathcal{D}$. However, $\mathcal{D}$ is shown to be decidable with access to a different physical theory $\mathrm{R}_2$. {\it Does this necessarily imply that $\mathrm{R}_2$ is an oracle theory?} Quite counter-intuitively, we demonstrate instances where $\mathrm{R}_2$ need not contain any oracle resource. In fact, we provide a simple but quite general logical structure to construct infinitely many such problems for any two distinct Turing theories. This may seem paradoxical at first glance since undecidable problems are not expected to be solved without access to an oracle. The paradox gets resolved by realizing that the Halting problem’s reducibility is not guaranteed when accessible resources differ with varying physical theories. Therefore, even though the resources of the latter theory make formerly undecidable problems decidable, they are still insufficient to solve $\mathcal{H}$ itself. Interestingly, we show that the results of Eisert {\it et al.} \cite{Eisert12} and Barry {\it et al.} \cite{Barry14} fit within our logical structure as two particular examples. Towards the end, we illustrate an explicit game involving  multiple agents (one player, two verifiers, and a referee) where the player having access to classical resources can always win the game only if she has some oracle resource. On the other hand, quantum resources yield a perfect winning strategy without requiring any oracle. Our results, thus, establish a previously unexplored fundamental relationship between undecidability and resource availability in different theories. 

\section{Preliminaries}
{\it Physical theory.--} Multiple frameworks have been proposed to analyse the operational features of arbitrary physical theories. For instance, within the framework of generalized probability models \cite{Hardy01,Barrett07,Chiribell11}, a theory is characterized by specifying the structure of allowed states, possible set of measurements ({\it i.e.} the structure of effect space), and the possible set of transformations. On the other hand, more recent development of the resource theory framework helps to investigate an agent’s operational capabilities under a set of allowed states and operations \cite{Chitambar19}. 
Following their general structure, we shall characterize a physical theory $\mathrm{T}$ by the tuple $\mathrm{T} \equiv \{\mathrm{B},\mathrm{F}\}$, where $\mathrm{B}$ is the set of available states and $\mathrm{F}$ is the set of allowed physical operations (transformations and measurements). Within the language of generalized probability theory \cite{Hardy01,Barrett07,Chiribell11}, classical and quantum theories can be seen as different physical theories. Additionally, this framework can be extended to any well-defined sets of states and operations (resources). Few definitions relevant to the computational aspect of a physical theory are in order.
\begin{definition}
[Distinctiveness] Two physical theories will be called distinct if some yes-no question $\mathbf{Q}$ is affirmative in one physical theory but negative in the other. 
\end{definition}
Operationally, this corresponds to some physical scenario possible in one theory but not in the other. For instance, any quantum advantage over the classical resources can be put forward as such a yes-no question with quantum theory having the affirmative answer. Note that, $\mathbf{Q}$ can be a meaningful question in the context of a physical theory only when it is defined using the resources (states and operations) of the theory (see Appendix-\ref{appen-a}). The query, "does any allowed operation transform $\rho$ to $\sigma$?", is a valid question across all physical theories containing states $\rho$ and $\sigma$, even though the solution can change across the theories.
\begin{definition}
[Turing theory] A physical theory will be called Turing theory if its computability strength is exactly the same as that of the Turing machine. 
\end{definition}
This means that, the computation by any Turing machine can be simulated by the states and operations of the Turing theory and vice-versa (see Appendix-\ref{appen-a}). 
Thus, questions undecidable to the Turing machine, such as $\mathcal{H}$, cannot be solved in any Turing theory. Both  classical theory and quantum theory are considered to be examples of Turing theories \cite{Cotogno03,Galton06,ARRIGHI12}.
\begin{definition}
[Oracle theory] A physical theory will be called an oracle theory if any computing model within its constructs can solve the Halting problem. 
\end{definition}
For instance, consider the accelerating Turing machine that can perform infinite computation in finite time by executing each step twice as fast as before \cite{Stewart91,Copeland02}; or the extended Turing machines that can store and compute arbitrary real numbers, something a Turing machine cannot \cite{Abramson71}.  Although they are purely imaginary and less appealing, we shall utilize their hypothetical structure in the discussion that follows and proceed to the next section containing the technical results.

\section{Results}
\begin{table}[t!]
\centering
\begin{tabular}{ cc }   
$\mathcal{D}\equiv\{D_i:=\mathbf{Q}\vee H_i\}$ ~~~~~~~~& $\tilde{\mathcal{D}}\equiv\{\tilde{D}_i:=\neg\mathbf{Q}\vee H_i\}$\\  
\begin{tabular}{ c|c|c } 
\hline
$\mathbf{Q}$ & $H_i$ & $D_i$\\
\hline\hline
~~~yes~~~ & ~~~yes ~~~&~~~ yes~~~ \\
\hline
~~yes~~ & ~~no ~~&~~ yes~~ \\
\hline
~~no~~ & ~~yes ~~&~~ yes~~ \\
\hline
~~no~~ & ~~no ~~&~~ no~~ \\
\hline
\end{tabular} ~~~~~~~~&  
\begin{tabular}{ c|c|c } 
\hline
$\neg\mathbf{Q}$ & $H_i$ & $\tilde{D}_i$\\
\hline\hline
~~~no~~~ & ~~~yes ~~~&~~~ yes~~~ \\
\hline
~~no~~ & ~~no ~~&~~ no~~ \\
\hline
~~yes~~ & ~~yes ~~&~~ yes~~ \\
\hline
~~yes~~ & ~~no ~~&~~yes~~ \\
\hline
\end{tabular} \\
\end{tabular}
\caption{Resource dependent undecidability: the problem $\mathcal{D}$ is trivially decidable in $\mathrm{T}_2$, while it boils down to solving the problem $\mathcal{H}$ in $\mathrm{T}_1$. On the other hand, the problem $\tilde{\mathcal{D}}$ is trivially decidable in $\mathrm{T}_1$, but boils down to $\mathcal{H}$ in $\mathrm{T}_2$. The statement $\mathrm{Q}$ is true in $\mathrm{T}_2$ but false in $\mathrm{T}_1$.}
\label{tab1}
\end{table}
\begin{theorem}\label{theo1}
For any pair of distinct Turing theories $\mathrm{T}_1$ and $\mathrm{T}_2$, there exists at least one decision problem, undecidable in $\mathrm{T}_1$ but trivially decidable in $\mathrm{T}_2$.
\end{theorem}
\begin{proof}
Let $\mathcal{H}\equiv\{H_i\}$ be the set of all questions of the halting problem with $i$ denoting the instances. $\mathcal{H}$ is undecidable in $\mathrm{T}_1$ and $\mathrm{T}_2$ since both of them are Turing theories. Furthermore, $\mathrm{T}_1$ and $\mathrm{T}_2$ being distinct, there exist some (at least one) yes-or-no problem $\mathbf{Q}$ whose answer is negative in $\mathrm{T}_1$ but  affirmative in $\mathrm{T}_2$. Now, consider a new decision problem $\mathcal{D}\equiv\{D_i\}$ whose questions are constructed from the conjunction of $\mathcal{H}$ and $\mathbf{Q}$. The instance $D_i$ of the problem $\mathcal{D}$ is true if and only if either $\mathbf{Q}$ or the $i^{th}$ instance of $\mathcal{H}$ is true, {\it i.e.} $D_i:=\mathbf{Q}\vee H_i,~\forall~i$. 
Since answer of $\mathbf{Q}$ is always negative in $\mathrm{T}_1$, the instance $D_i$ can have an affirmative answer in $\mathrm{T}_1$ only if $H_i$ is affirmative in $\mathrm{T}_1$; and accordingly the problem $\mathcal{D}$ boils down to the problem $\mathcal{H}$ and is hence undecidable. On the other hand, $\mathcal{D}$ is trivially decidable in $\mathrm{T}_2$ since $\mathbf{Q}$ is always true there, although $\mathcal{H}$ is still undecidable in $\mathrm{T}_2$.
\end{proof}
Note that the role of $\mathrm{T}_1$ and $\mathrm{T}_2$ in Theorem \ref{theo1} can be interchanged just by replacing $\mathbf{Q}$ with $\neg\mathbf{Q}$. This assertion is made clear in Table \ref{tab1}. A limitation of Theorem \ref{theo1} is that the problem $\mathcal{D}$ is trivially decidable in $\mathrm{T}_2$, {\it i.e.} the solution remains the same for all instances (affirmative in our case). Our next result overcomes this limitation.
\begin{theorem}\label{theo2}
For any pair of distinct Turing theories $\mathrm{T}_1$ and $\mathrm{T}_2$, there exist infinitely many decision problems undecidable in $\mathrm{T}_1$ but non-trivially decidable in $\mathrm{T}_2$.
\end{theorem}
\begin{proof}
Let ${\mathcal{H}_f}\equiv\{{H}_{f(i)}\}$ be any non-trivially decidable subset of $\mathcal{H}$, where ${f}$ is a surjective function mapping the halting problem instances $\{i\}$ to the constructed problem’s instances \{$f(i)$\}. Now, consider the decision problem ${\mathcal{D}[f]}\equiv\{{D}_i[f]\}$, where,
\begin{equation} \label{eu1}
{D}_i[f]:=\left(\neg\mathbf{Q} \, \wedge \, H_i\right)~\vee~\left(\mathbf{Q} \, \wedge \, H_{f(i)}\right),~\forall~i.
\end{equation}
Once again, since $\mathbf{Q}$ is false in $\mathrm{T}_1$ but true in $\mathrm{T}_2$, the problem $\mathcal{D}[f]$ boils down to $\mathcal{H}$ in $\mathrm{T}_1$, whereas it boils down to the non-trivially decidable problem ${\mathcal{H}_f}$ in $\mathrm{T}_2$. Varying the surjective mapping $f$, infinitely many such decision problems can be constructed. 
\end{proof}
The above theorems, thus, provide  an illuminating picture regarding the computability landscape of distinct Turing theories (see Fig.\ref{fig1}). What we will show next is that the recent results of Eisert {\it et al.} \cite{Eisert12} and Barry {\it et al.} \cite{Barry14} can be seen as special instances of the above theorem with $\mathrm{T}_1$ being the quantum theory and $\mathrm{T}_2$ the classical theory.
\begin{corollary}
The undecidability-decidability status of the QMOP-CMOP analogues and the QOMDP-POMDP goal state reachability follow the structure of Theorem \ref{theo2}.
\end{corollary}
\begin{proof}
The QMOP and CMOP ask whether any finite outcome sequence is impossible for a given measurement device (quantum and classical, respectively). To prove their results, the authors in \cite{Eisert12} ask the question `does the measurement prohibit destructive interference?'. While quantum measurements yield a negative answer to this question, the classical scenario gives an affirmative one. Then the proof proceeds by showing that a negative answer implies the reducibility of a known undecidable problem -- the matrix mortality problem (MMP) -- to the measurement occurrence problem (MOP). On the other hand, an affirmative solution equates the decidable nonnegative-MMP to the MOP. In the present context, the question they posed plays the role of problem $\mathbf{Q}$ in Eq.(\ref{eu1}). Additionally, the undecidability of MMP follows from the undecidability of $\mathcal{H}$ \cite{Halava01}, and nonnegative-MMP is non-trivially decidable in both settings \cite{Blondel97}. Therefore, the QMOP-CMOP analogues are one among the infinite decision problems generated by the structure of Theorem \ref{theo2}. However, the absence of an appropriate surjective function $f$ restrains the authors in \cite{Eisert12} to consider analogous problems only instead of the same problem. Similar arguments hold for the QOMDP-POMDP goal state reachabilities since the former’s undecidability was shown by reduction from QMOP, and the latter’s decidability eventually stems from the associated stochastic transition matrix elements being non-negative \cite{Barry14}.
\end{proof}
\begin{figure}[t!]
\centering
\includegraphics[width=\linewidth]{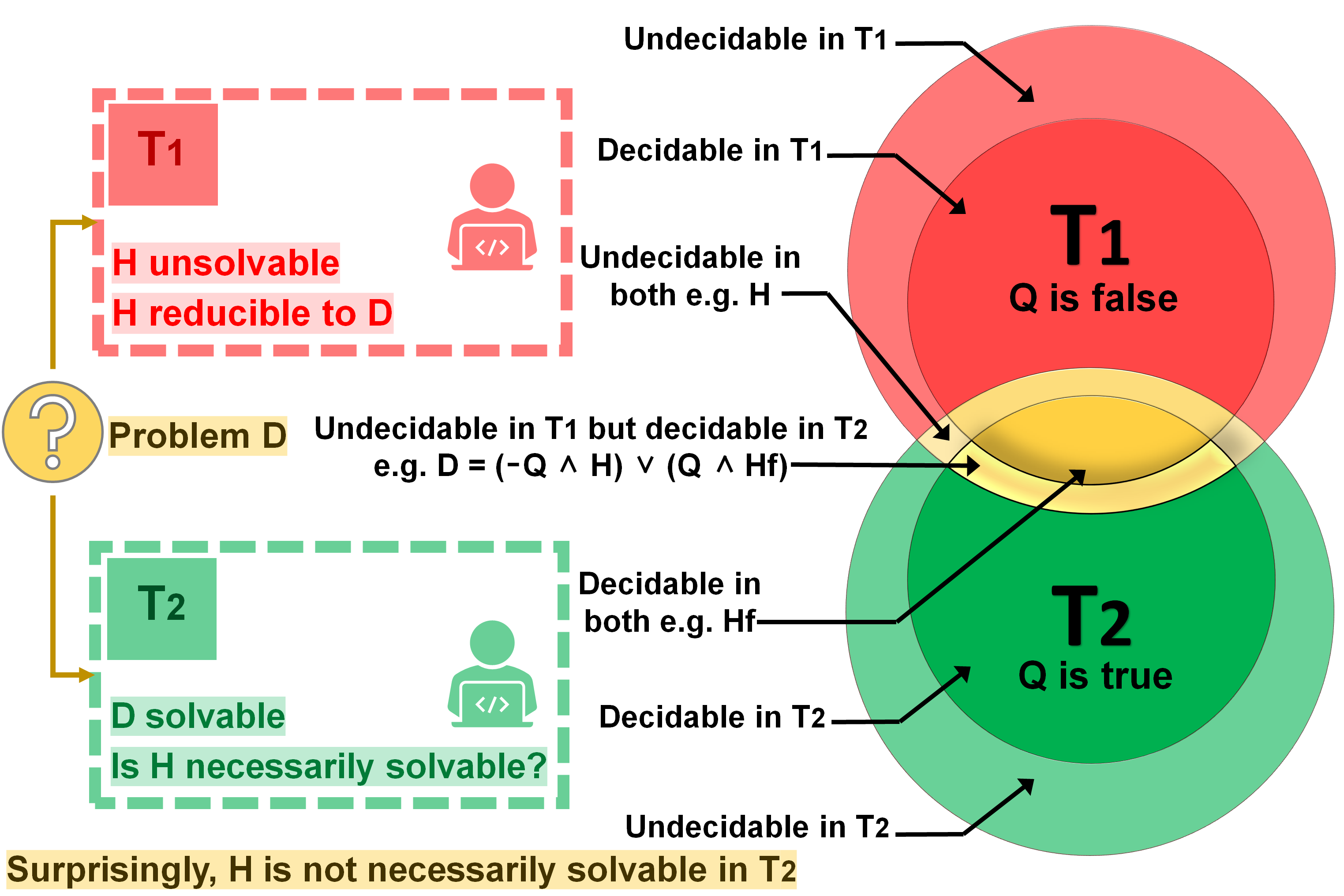}
\caption{(Color online) Computability landscape of distinct Turing theories: Decision problem $\mathcal{D}$ is undecidable in $\mathrm{T}_1$ but decidable in $\mathrm{T}_2$, although the Halting problem $\mathcal{H}$ remains undecidable in both. $\mathcal{H}_f$ is some non-trivially decidable subset of $\mathcal{H}$.}.
\label{fig1}
\end{figure}
We now go forward to show that our results can be extended beyond the confines of mathematical logic to physical scenarios. With this aim, we introduce a game named the Physical Post correspondence Game (PPCG). Its highly general rules permit the player to use both classical and quantum theories to play. Interestingly, we show that the winning strategy becomes undecidable with classical resources whereas it is perfectly decidable in quantum theory. 

{\it PPCG.--}The game involves a player, say Alice, two verifiers $V_{1}$ and $V_{2}$ prohibited from communicating with each other except in ways the rules of the game allow, and a referee. Verifier $V_{2}$ randomly generates a PCP instance and sends it to Alice. Post correspondence problem (PCP) is a familiar undecidable problem introduced in 1946 by Emil L. Post \cite{Post1946}. A PCP instance consists of some finite collection of dominoes (named \textit{\(A_1\)}, \textit{\(A_2\)} etc.) with a numerator and denominator string for each. The player should find whether the dominoes can be arranged such that the concatenated strings in the numerator and denominator become equal. Here, the strings are made over the alphabet $\{1, 2, 3, 4\}$. Upon finding a match, Alice sends it to referee; otherwise, she sends the string “no match”. 
\begin{table}[t!]
\caption{RULES of PPCG}\label{tab2}
\centering
\begin{tabular}{c}
\hline\hline
\hspace{-1.8cm}{\bf Step-$1$}:\hspace{1.3cm}{\bf Verification of devices by $V_1$}\\
\hline\hline
(i) For each domino $\textit{\(A_i\)}$, Alice approaches  $V_{1}$ with $N^{\textcolor{blue}{\star}}$\\
boxes, a device $\mathrm{M}_{ea}$ and a device $\mathrm{M}_{ix}$.\\
\hline
(ii) $V_{1}$ verifies the devices using $1/2$ of the boxes chosen\\
randomly$^{\textcolor{blue}{\mathsection}}$ and discards these boxes.\\
Other boxes are acted upon by $\mathrm{M}_{ix}$ and labelled “\textit{mixed}”.\\
\hline\hline
\hspace{-2.7cm}{\bf Step-$2$}: \hspace{1.7cm}{\bf Encoding by Alice}\\
\hline\hline
Alice selects $1/4$-th of the boxes encoding domino $\textit{\(A_i\)}$.\\
To aid this, she is free to measure any $"\textit{mixed}"$\\
boxes and relabel them with the outcomes $(\mathtt{hh}, \mathtt{ht},\cdots)^{\textcolor{blue}{\$}}$.\\
\hline\hline
\hspace{-1.9cm}{\bf Step-$3$}: \hspace{0.95cm}{\bf Probability verification by $V_1$}\\
\hline\hline
$V_{1}$ counts number $(n)$ of $\mathtt{hh}, \mathtt{ht}, \mathtt{th}~\&~ \mathtt{tt}$ of encoded boxes$^{\textcolor{blue}{\#}}$.\\
If $n(\mathtt{hh})\times n(\mathtt{tt}) = n(\mathtt{ht})\times n(\mathtt{th})$, the\\
devices and the encoded boxes are forwarded to $V_{2}$.\\
\hline\hline
\hspace{-2.05cm}{\bf Step-$4$}: \hspace{1.05cm}{\bf Encoding verification by $V_2$} \\
\hline\hline
$V_{2}$ checks encoding for numerator and denominator\\ 
using randomly chosen $1/3$ of the boxes for each$^{\textcolor{blue}{@}}$.\\
Remaining boxes are sent to referee along\\
with instruction “$\mathrm{L}$” if $k_i\leq q_i$, else with instruction “$\mathrm{R}$”$^{\textcolor{blue}{\mathparagraph}}$.\\
\hline\hline
\hspace{-1.15cm}{\bf Step-$5$}: \hspace{0.3cm}{\bf Decoding of domino strings by referee}\\
\hline\hline
Upon receiving “$\mathrm{L}$”, referee decodes the numerator first,\\
else the denominator$^{\textcolor{blue}{\dagger}}$. Boxes with outcome $\mathtt{t}$\\
are discarded. Then the other string is decoded\\
from the other compartments of the remaining boxes.\\
\hline\hline
\textcolor{blue}{$\star$} $N$ is sufficiently large for performing all steps [see~~~~\\
Appendix-\ref{appen-b} for more details].\\
\textcolor{blue}{$\mathsection$} By repeated measurements, acting $\mathrm{M}_{ix}$ on measured\\
~~~~coins and remeasuring them, arbitrarily many times.\\
~\textcolor{blue}{$\$$} Any domino can be encoded by judiciously choosing\\
a quarter of the mixed boxes. After
encoding, the~~\\
remaining boxes are discarded [see Appendix-\ref{appen-c}].~~\\
\textcolor{blue}{$\#$} Label "mixed" is counted as $1/4$-th of each outcome.\\
\textcolor{blue}{$@$} As $N$ is large enough, $1/3$-rd of boxes are sufficient.\\
~\textcolor{blue}{$\mathparagraph$} $V_{2}$ generated PCP instance and knows $k_{i}$ \& $q_{i}$ values.\\
\textcolor{blue}{$\dagger$} Respective compartments of all boxes get measured.\\
\hline
\end{tabular}
\end{table}

The game will be played using certain physical devices whose
descriptions and properties are in order. A coin is an object yielding two distinct outcomes, say head ($\mathtt{h}$) or tail ($\mathtt{t}$), when measured by a measurement device $\mathrm{M}_{ea}$. The outcome remains constant upon re-measurement(s). If a measured coin is subjected to a mixing device $\mathrm{M}_{ix}$ and re-measured, both outcomes are equally likely. There is another device called box $(\mathrm{B})$ having left $(\mathrm{L})$ and right $(\mathrm{R})$ compartments each capable of containing a coin. A box is said to be subjected to a mixing device $\mathrm{M}_{ix}$ if both of its coins are acted upon by the device. 
A set of boxes can be used to encode a domino $A_i$ in the following manner: First, the numerator and denominator strings (say, "$121$" and "$34$") get converted to their probability values ($0.121$ and $0.34$) termed $k_i$ and $q_i$, respectively. Then, the boxes should be chosen such that the measurements of their $(\mathrm{L})$ and $(\mathrm{R})$ compartments give outcome $\mathtt{h}$ with the corresponding probabilities $k_i$ and $q_i$. Thus, given the encoded boxes, the strings can be decoded by measuring the compartments. The game is played following a set of rules outlined step-wise in Table \ref{tab2}, which get repeated for each domino $A_i$. Each domino of the PCP instance gets encoded into a set of boxes which undergo a series of verifications before reaching the referee who decodes them. Alice wins if the referee finds the matching order she provided to be valid for the decoded dominoes. If the solution was “no match”, the referee runs a brute-force match-finding algorithm for an arbitrary amount of time. Again, Alice wins if no match is found. The natural question that comes up is, given some PCP instance, does there exist any strategy guaranteeing the player’s victory in PPCG? In the following, we investigate the decidability status of this problem for a player with access to classical/ quantum resources.

{\it Undecidability in classical scenario.--} We start by examine whether a classical player can always ensure victory through some strategy. Denoting coin states as probability vectors, action of a classical mixing device $\mathrm{M}^{c}_{ix}$ can be represented as a stochastic matrix {\footnotesize$\begin{pmatrix}1/2 & 1/2 \\1/2 & 1/2\end{pmatrix}$} (see Appendix-\ref{appen-c} for more details). All the four outcomes obtained from a box subjected to $\mathrm{M}^c_{ix}$ are equally likely  (step-$1$), {\it i.e.}, it eliminates any correlations between the left and right coins in a box. After Alice's encoding, the probability verification in step-$3$ ensures no correlation between the coins. This, after the encoding verification by $V_2$, allows the referee to decode the domino strings accurately in step-$5$ even though he selectively discards some of the boxes depending on measurement outcomes of one compartment. In short, a classical player lacks any perfect strategy to win using some erroneous encoding. The only way left to always win is to be able to solve every instance of the undecidable PCP. Thus, PPCG problem becomes undecidable for a classical player.

{\it Quantum winning strategy.--}For a player having access to quantum resources, a two level quantum system (qubit) serves the purpose of the coin as required in the game. Assigning the role of $\mathtt{h}$ \& $\mathtt{t}$ to the computational basis ({\it i.e.} eigenstates of Pauli $\sigma_z$ operator), the action of measurement device is fulfilled by the projective measurement along $z$-direction, {\it  i.e.} $\mathrm{M}^q_{ea}\equiv\sigma_z$. The operator {\footnotesize$\begin{pmatrix}1/\sqrt{2} & -1/\sqrt{2} \\1/\sqrt{2} & ~~1/\sqrt{2}\end{pmatrix}$} accomplishes the demand  of the mixing device $\mathrm{M}^q_{ix}$ in this case. Note that, the $\mathrm{M}^q_{ix}$ satisfying the verification in step-$1$ is not unique (see Appendix-\ref{appen-d}). The player chooses this particular one to ensure her perfect winning. A generic quantum state for a box (containing two coins) is $a\ket{\mathtt{hh}}+b\ket{\mathtt{ht}}+c\ket{\mathtt{th}}+d\ket{\mathtt{tt}}\equiv(a,b,c,d)^{\mathsf{T}}$, where $a,b,c,d\in\mathbb{C}$ with $|a|^2+|b|^2+|c|^2+|d|^2=1$ and $\mathsf{T}$ denotes transposition in computational basis. One of the winning strategies for Alice involves sending the matching order “\textit{\(A_{1}\)}” as the solution to the referee regardless of the PCP instance she receives. She prepares all the boxes corresponding to domino \textit{\(A_{1}\)} in the normalized state $\ket{\psi}:= 1/2(a+b+c+d,-a+b-c+d,-a-b+c+d,a-b-c+d)^{\mathsf{T}}$. Upon the action of $\mathrm{M}^q_{ix}\otimes\mathrm{M}^q_{ix}$ in step-$1(ii)$, this state transforms into $\ket{\phi}:=(a,b,c,d)^{\mathsf{T}}$. Note that, unlike the classical case, the outcomes are not equally likely here since the quantum mixing device allows interference. For encoding, Alice chooses $1/4$ of the boxes all in state $\ket{\phi}$ without doing any further measurement or relabelling. Evidently, these boxes qualify the probability verification. The box state $\ket{\phi}$ passes the verification in step-$4$ and the referee's verification step at the end if the following conditions are satisfied: $|a|=l,~|b|=\sqrt{k_1-l^2},~|c|=\sqrt{q_1-l^2},~\&~|d|=\sqrt{1-k_1-q_1+l^2}$, where $l=k_1$ if $k_1\leq q_1$, else $l=q_1$. Infinitely many solutions are possible satisfying these conditions (see Appendix-\ref{appen-d}). The remaining dominoes $(A_2,A_3,\cdots)$ are classically encoded. This does not hinder the winning strategy as they do not appear in the matching order sent to the referee. Thus, PPCG problem is trivially decidable in quantum scenario. Manifestly, this undecidability/decidability status in classical/quantum theory follows the structure of Theorem \ref{theo1}, with problem $\mathbf{Q}$ enquiring whether interference is allowed.     

\section{Discussion} 
Interestingly, PPCG contains certain rudimentary elements of the well-known zero-knowledge proofs \cite{Barak2009} where a prover aims to convince the verifier(s) of possessing certain knowledge without revealing it. In our case, verifiers with knowledge of only classical theory are convinced that a halting problem oracle is necessary to win always. From their perspective, a player can convince them of possessing the oracle by winning consistently. The probability that a player without an oracle never loses reduces with subsequent rounds. Also, no additional information gets revealed. However, our game is not exactly at par to the standard zero-knowledge proofs \cite{Blum1988,Sahai2003,Wu2014} as it requires a trusted referee among others. Nonetheless, it shares a fundamental similarity, a dishonest classical prover (player) has no perfect winning strategy. However, a prover with quantum resources can successfully deceive the classical verifiers that she possesses a halting problem oracle.

To summarize, in the classical scenario, the multiple PPCG verifications collectively ensure that the encoding is correct, so the classical player cannot ensure victory without knowing the PCP solution. The quantum player, however, can successfully trick the verifiers with a wrong encoding and win. Interestingly, even if $V_1$ and $V_2$ communicate and the deception gets revealed, Alice is still declared the winner as per the game rules, as she has passed all the verifications. Thus, the game remains classically-undecidable and quantum-decidable even if the verifiers are allowed to communicate. However, then the verifiers with knowledge of classical theory alone will realise that the quantum player is winning not by an oracle but by some physical theory unknown to them. A similar consequence occurs if the referee is replaced by a third verifier $V_3$ who will sense the deception when the solution $A_{1}$ gets repeated. Thus, by prohibiting the verifiers from communicating and including an impartial referee who simply follows the rules, the game retains its classical undecidability, quantum decidability and zero-knowledge character (to verifiers knowing only classical theory) but allows the quantum player to trick all the verifiers into believing that she owns an oracle.

\section{Conclusions}
Several significant problems in quantum computing and in quantum many-body physics are likely undecidable. Although few such results have been established in recent times, proving undecidability, unlike decidability, is generally laborious. Our results can help screen the likely candidates by examining their classical counterparts’ decidability and statements of distinctiveness.
Our highly general results establish deep foundational links between physics and computability. In this era of information and computation theories redefining physics, with explosive advances gained by investigating such links in the complexity sector, we expect that our results concerning its relatively less explored sister sector computability will prove beneficial.

{\bf Acknowledgement:} I convey my sincere gratitude to Dr. Manik Banik (IISER-Thiruvananthapuram) for all the discussions we had, for his time and efforts in reviewing and editing the manuscript, and most of all, for the immense motivation he provided as my supervisor. Dr. Mir Alimuddin, Edwin Lobo, Sahil Gopalkrishna Naik, Samrat Sen, and Ramkrishna Patra (IISER-Thiruvananthapuram) are all acknowledged for the productive discussions with them. I thank Dr. Christian Gogolin (Covestro Deutschland AG) and Prof. Markus P. Müller (IQOQI, Vienna) for their helpful suggestions.

\newpage
\appendix

\section{Physical theories and computation}\label{appen-a}
A question can be said to be meaningful within the context of some physical theory if it is some yes-or-no query defined using the states and operations of the theory and additional logical structures. For instance, the query, "does any allowed physical operation transform $\rho$ to $\sigma$?", is a valid question across all physical theories that contain states $\rho$ and $\sigma$. However, the solution to this can change depending on the operations available in each theory. Note that, valid questions can be defined more generally without having to specify any particular state or operation. For instance, the query, "given the rules of Bell CHSH game, does any strategy utilizing the available states and operations (allowed in a theory) offer a success probability more than $75\%$?", is a well-defined question for which the solution changes across classical and quantum theories.

A computational problem can be also viewed as a valid question posed in some physical theory provided the Turing machine required to solve it can be simulated by the theory. A physical theory is said to simulate some Turing machine if the following hold: Each input string and final state (accept or reject states) of the Turing machine is uniquely mapped to some subset of the available states in the physical theory, and each available state mapped to an input string can be transformed to the available state mapped to the corresponding final state by some finite protocol. Here, finite protocol refers to a finite number of steps where each step involves some finite number of available states, physical operations and additional logical operations. Thus, a computational problem can be said to be equivalent to a valid question regarding the transformability across states of some physical theory, provided the theory can simulate the associated Turing machine.

\section{Number of boxes required}\label{appen-b}
For the PPCG game, the number of boxes required to properly encode each domino ($N$) depends on the longest domino string in the PCP instance. Consider some string (say "$133$") to be encoded into a set of boxes. The uncertainty in the probability measured while decoding the string can be made arbitrarily small by increasing the number of boxes in the set. Let $\mathcal{N}(l)$ denote the minimum number of boxes required for reducing the uncertainty to $l$ decimal digits. This means that, using $\mathcal{N}(4)$ boxes to encode the string "$133$" will help the decoder to obtain the probability value $0.1330$ upon measurements. Terminating the string before '$0$' (since the strings are made over $\{1,2,3,4\}$), the domino string gets decoded accurately. Therefore, if $l_{max}$ is the length of the longest domino string in the whole PCP instance, the minimum number of boxes required for the referee to decode both the strings of any domino goes to $N^\prime= \mathcal{N}(l_{max}+1)/p_{min}$. Note that, $p_{min}$ is the minimum among all values of $k_i$ and $q_i$. It is included to account for the referee selectively discarding certain boxes in step-$5$ depending on measurement outcomes of one compartment. Scaling this up to step-$1$ in Table \ref{tab2}, we find that Alice should begin with a total of $N=24\times N^\prime$ boxes for the referee to end up with $N^\prime$ boxes by step-$5$ and decode every string accurately.

\section{PPCG in the classical scenario}\label{appen-c}
A generic state of a two-level classical coin $\mathcal{C}_p$ can be represented as a probability vector $\mathcal{C}_p:=p\mathcal{C}(\mathtt{h})+(1-p)\mathcal{C}(\mathtt{t})$, where $\mathcal{C}(\mathtt{h}):=(1,0)^\mathsf{T},~\mathcal{C}(\mathtt{t}):=(0,1)^\mathsf{T}$, and $p\in[0,1]$. The classical measurement device $\mathrm{M}^c_{ea}$ can be seen as a probability update rule: upon obtaining the outcome $\mathtt{h}~[\mathtt{t}]$ on some state $\mathcal{C}_p$ the updated coin state becomes $\mathcal{C}(\mathtt{h})~[\mathcal{C}(\mathtt{t})]$, {\it i.e.} 
\begin{subequations}
\begin{align}
\mathcal{C}_p\xrightarrow{\{\mathrm{M}^c_{ea},\mathtt{h}\}}\mathcal{C}(\mathtt{h}),~\mbox{when}~ p\in(0,1];\\
\mathcal{C}_p\xrightarrow{\{\mathrm{M}^c_{ea},\mathtt{t}\}}\mathcal{C}(\mathtt{t}),~\mbox{when}~ p\in[0,1).
\end{align}
\end{subequations}
Since a mixing device $\mathrm{M}^c_{ix}$ maps both the coin states $\mathcal{C}(\mathtt{h})$ and $\mathcal{C}(\mathtt{t})$ to the state $\mathcal{C}_{1/2}$, it can be represented as a stochastic matrix.
\begin{align}
\mathrm{M}^c_{ix}\equiv \frac{1}{2}\begin{pmatrix}1 & 1 \\1 & 1\end{pmatrix}.    
\end{align}

Infact, $\mathrm{M}^c_{ix}$ maps any $\mathcal{C}_p$ to the state $\mathcal{C}_{1/2}$ for $p\in[0,1]$. A box state $\mathcal{B}[\alpha,\beta,\gamma]$ can be thought as a two-coin state and can be represented as 
\begin{align}
\mathcal{B}[\alpha,\beta,\gamma]&\equiv\alpha\mathcal{C}(\mathtt{hh})+\beta\mathcal{C}(\mathtt{ht})+\gamma\mathcal{C}(\mathtt{th})\nonumber\\
&~~~~~~~+(1-\alpha-\beta-\gamma)\mathcal{C}(\mathtt{tt}),\\ \mbox{where},&~~ 0\le\alpha,\beta,\gamma,\alpha+\beta+\gamma\le1;\nonumber\\
\mathcal{C}(\mathtt{hh})&=(1,0,0,0)^\mathsf{T},~~~
\mathcal{C}(\mathtt{ht})=(0,1,0,0)^\mathsf{T},\nonumber\\
\mathcal{C}(\mathtt{th})&=(0,0,1,0)^\mathsf{T},~~~
\mathcal{C}(\mathtt{tt})=(0,0,0,1)^\mathsf{T}.\nonumber
\end{align}
Such a state generally allows correlation between two coins placed at the left and right compartments of the Box \cite{Guha2021}. However, action of the mixing device $\mathrm{M}^c_{ix}$ on both compartments kills any such correlation and the box state becomes completely mixed, {\it i.e.}
\begin{align}
\mathrm{M}^c_{ix}\otimes \mathrm{M}^c_{ix}\left(\mathcal{B}[\alpha,\beta,\gamma]\right)=\mathcal{C}^\mathrm{L}_{1/2}\otimes\mathcal{C}^\mathrm{R}_{1/2}. \label{eqb4}
\end{align}
Here, superscripts are used to indicate the state of the coins in two different compartments. For a box state $\mathcal{B}[\alpha,\beta,\gamma]$, the states of the coins in its left and right compartments read as
\begin{align*}
\mathcal{C}^\mathrm{L}_p=p\mathcal{C}(\mathtt{h})+(1-p)\mathcal{C}(\mathtt{t}),~\mbox{where},~p:= \alpha+\beta;\\
\mathcal{C}^\mathrm{R}_q=q\mathcal{C}(\mathtt{h})+(1-q)\mathcal{C}(\mathtt{t}),~\mbox{where},~q:= \alpha+\gamma.
\end{align*}
If the box state $\mathcal{B}[\alpha,\beta,\gamma]$ lacks any correlation between its left and right compartment coins, it can be written in the separable form $\mathcal{B}[\alpha,\beta,\gamma]$ = $\mathcal{C}^\mathrm{L}_p\otimes\mathcal{C}^\mathrm{R}_q$, which implies
\begin{align*}
\alpha&=(\alpha+\beta)(\alpha+\gamma),\\
\beta&=(\alpha+\beta)(1-\alpha-\gamma),\\
\gamma&=(1-\alpha-\beta)(\alpha+\gamma),\\
1&-\alpha-\beta-\gamma=(1-\alpha-\beta)(1-\alpha-\gamma).
\end{align*}
The above conditions hold if and only if $\alpha(1-\alpha-\beta-\gamma)= \beta\gamma$ which gets checked during step-$3$ of PPCG. In classical scenario, this ensures no correlation between the coins and hence prohibits any cheating strategies by the player.\\\\
{\bf Remarks:}
\begin{itemize}
\item Note that, after the action of mixing device on each compartment of a box, the box state becomes completely mixed [see Eq.(\ref{eqb4})]. So, only $1/4$-th of the mixed boxes can give any particular outcome ($\mathtt{hh}, \mathtt{ht}, \mathtt{th}$ or $\mathtt{tt}$). Therefore, in step-$2$, the required number of encoded boxes are limited to $1/4$-th of the mixed boxes to ensure that the number of boxes needed for each outcome to encode any arbitrary domino is available.
\item Since the state of a mixed box is completely mixed, during counting in step-$3$, any "mixed" label box contributes $1/4$-th to each of the outcomes.
\item Note that, even though Alice cannot cheat with the encoding, she can still send the solution "no match" when the solution is unknown to her. However, even this does not ensure her victory as the referee is given an arbitrary amount of time to try finding a match.
\end{itemize}

\section{PPCG in the quantum scenario}\label{appen-d}
In the quantum scenario, a quantum coin is a two-level quantum system (qubit) and its state can be represented as $\ket{\phi}=\sqrt{p}\ket{\mathtt{h}}+e^{\mathtt{i}\zeta} \sqrt{1-p}\ket{\mathtt{t}}\in\mathbb{C}^2$, where $p\in[0,1]$ and $\zeta\in[0,2\pi]$. Without any loss of generality we will consider the computational basis ({\it i.e.} eigenstates of Pauli $\sigma_z$ operators) as our coin basis, {\it i.e.} $\ket{\mathtt{h}}\equiv\ket{0}$ and $\ket{\mathtt{t}}\equiv\ket{1}$. Clearly, $\mathrm{M}^q_{ea}\equiv\sigma_z$ satisfies the requirements of the measurement device. Consider an operator of the following form:
\begin{align}
\mathrm{M}^q_{ix}[\chi]:=\frac{1}{\sqrt{2}}\begin{pmatrix}1 & -e^{\mathtt{i}\chi} \\e^{\mathtt{i}\chi} & 1\end{pmatrix};~~~\chi\in[0,2\pi].
\end{align}
Action of this on the state $\ket{0}$ and $\ket{1}$ yields
\begin{subequations}
\begin{align}
\mathrm{M}^q_{ix}[\chi]\ket{0}&=\frac{1}{\sqrt{2}}(\ket{0}+e^{\mathtt{i}\chi}\ket{1}),\\
\mathrm{M}^q_{ix}[\chi]\ket{1}&=\frac{1}{\sqrt{2}}(-e^{\mathtt{i}\chi}\ket{0}+\ket{1}).
\end{align}
\end{subequations}
Further measurement of $\sigma_z$ operator ({\it i.e.} quantum measurement device $\mathrm{M}^q_{ea}$) on these evolved states yields outcomes $\mathtt{h}$ and $\mathtt{t}$ with equal probability. Thus any of these parametric family of operators suffice as a quantum mixing device $\mathrm{M}^q_{ix}$. In the protocol we have considered $\chi=0$.

A generic box state in this case read as
\begin{align}
\ket{\mathcal{B}}=w\ket{\mathtt{hh}}+x\ket{\mathtt{ht}}+y\ket{\mathtt{th}}+z\ket{\mathtt{tt}}\in(\mathbb{C}^2)^{\otimes2};\\
w,x,y,z\in\mathbb{C}~\&~|w|^2+|x|^2+|y|^2+|z|^2=1.\nonumber
\end{align}
Under the action of $\mathrm{M}^q_{ix}$ on both of the coins we get
\begin{align}
&\left(\mathrm{M}^q_{ix}\otimes\mathrm{M}^q_{ix}\right)\ket{\mathcal{B}} =\frac{1}{2} \times\nonumber\\
&~[(w-x-y+z)\ket{\mathtt{hh}}+(w+x-y-z)\ket{\mathtt{ht}}\nonumber\\  &+(w-x+y-z)\ket{\mathtt{th}}+(w+x+y+z)\ket{\mathtt{tt}}].
\end{align}
Upon measuring this evolved box state in $\sigma_z$ basis on left and right coins, the probabilities for different outcomes are given by
\begin{subequations}
\begin{align}
p(\mathtt{hh})&=\frac{1}{4}\times|w-x-y+z|^2,\\
p(\mathtt{ht})&=\frac{1}{4}\times|w+x-y-z|^2,\\
p(\mathtt{th})&=\frac{1}{4}\times|w-x+y-z|^2.\\
p(\mathtt{tt})&=\frac{1}{4}\times|w+x+y+z|^2.
\end{align}
\end{subequations}
Note that, unlike the classical case [see Eq.(\ref{eqb4})], here the probabilities are not necessarily equal due to quantum interference. This makes a quantum player powerful enough to cheat with the encoding and thus ensure victory despite the fact that she passes both the device and probability verifications (step-$1$ and step-$3$). In our protocol, by the end of step-$3$, the encoded boxes for domino $\textit{\(A_1\)}$ are in the state $\ket{\phi}:=(a,b,c,d)^{\mathsf{T}}$. For this state to qualify the encoding verification next (step-$4$), $V_2$ should obtain outcome $\mathtt{h}$ upon measurement of the left and right compartments with probabilities $k_1$ and $q_1$, respectively. Hence, the following conditions must be satisfied:
\begin{align}
|a|^2 + |b|^2 = k_1~~~\&~~~
|a|^2 + |c|^2 = q_1.
\label{eqnc6}
\end{align}
Recall that verifier $V_2$ never measures both compartments of the same box. The conditions above are derived accordingly. However, this is not the case in step-$5$, where the referee measures both compartments of certain boxes. To pass this final step, the decoded numerator and denominator strings of domino $\textit{\(A_1\)}$ should be equal. To obtain such a result, the following conditions should hold.
\begin{subequations}
\begin{align}
&|a|^2 + |b|^2 = \frac{|a|^2}{|a|^2 + |b|^2}, ~\mbox{ if }~ k_1\leq{q_1}
\label{eqnc7a}\\
&|a|^2 + |c|^2 = \frac{|a|^2}{|a|^2 + |c|^2}, ~\mbox{ otherwise. }
\label{eqnc7b}
\end{align}
\end{subequations}
Note that, Eqs.(\ref{eqnc7a})-(\ref{eqnc7b}) correspond to the cases where the referee measures the left or the right compartments first. The left hand sides of both equations denote the probabilities for obtaining outcome $\mathtt{h}$ in the compartment first measured, whereas the right hand sides denote the same for the remaining compartment after those boxes with $\mathtt{t}$ outcomes get discarded [see step-$5$]. Solving Eq.(\ref{eqnc6}), Eq.(\ref{eqnc7a}), and Eq.(\ref{eqnc7b}) and using the normalization condition $|a|^2 + |b|^2 + |c|^2 + |d|^2 = 1$, we have
\begin{subequations}
\begin{align}
&|a|=l, ~~|b|=\sqrt{k_1-l^2}, ~~|c|=\sqrt{q_1-l^2},
\label{eqnc8a}\\
&~~~~~~~~~~~~~~~~~|d|=\sqrt{1-k_1-q_1+l^2}
\label{eqnc8b},\\
&\mbox{where, }~l=k_1 ~\mbox{if }~k_1\leq{q_1}, ~\mbox{else }~l=q_1. \nonumber
\end{align}
\end{subequations}
It is straightforward to see that infinitely many solutions are possible satisfying the first three conditions [in Eq.(\ref{eqnc8a})]. The same holds while considering the remaining equation also [Eq.(\ref{eqnc8b})] since the strings are made over the alphabet $\{1,2,3,4\}$ which ensures that $0 < $\{$k_1$, $q_1$\}$ < 0.5$, and the expression ($1-k_1-q_1+l^2$) remains non-negative. Thus, by preparing the boxes for domino $\textit{\(A_1\)}$ in the appropriate state $\ket{\psi}$, Alice can always ensure victory in the quantum scenario.
\bibliography{Manuscript}

\end{document}